\documentclass{article}

\usepackage{arxiv}          
\usepackage[utf8]{inputenc} 
\usepackage[T1]{fontenc}    
\usepackage{hyperref}       
\usepackage{url}            
\usepackage{booktabs}       
\usepackage{amsfonts}       
\usepackage{microtype}      
\usepackage{amsmath}        
\usepackage{amsthm}         
\usepackage[normalem]{ulem} 
\usepackage{graphicx}
\usepackage{natbib}
\usepackage{placeins}
\usepackage{color}

\DeclareMathOperator*{\argmax}{arg\,max}
\DeclareMathOperator*{\argmin}{arg\,min}
\newtheorem{prop}{Proposition}

\title{Guidance on Individualized Treatment Rule Estimation in High Dimensions}

\date{}

\author{
  Philippe Boileau \\
  Department of Epidemiology, Biostatistics \\
  and Occupational Health, \\
  Department of Medicine,\\
  McGill University\\
  \texttt{philippe.boileau@mcgill.ca} \\
  \And
  Ning Leng \\
  Genentech Inc. \\
  \texttt{leng.ning@gene.com} \\
  \And
  Sandrine Dudoit \\
  Department of Statistics, \\
  Division of Biostatistics,\\
  Center for Computational Biology,\\
  University of California, Berkeley\\
  \texttt{sandrine@stat.berkeley.edu} \\
}

\renewcommand{\headeright}{}
\renewcommand{\undertitle}{}
\renewcommand{\shorttitle}{}

\newcommand{\reviewer}[1]{#1}

\newcommand\startappendix{%
    \makeatletter
       \setcounter{table}{0}
       \renewcommand{\thetable}{A\arabic{table}}
       \setcounter{figure}{0}
       \renewcommand{\thefigure}{A\arabic{figure}}
       \setcounter{section}{0}
       \renewcommand{\thesection}{A\arabic{section}}
       \setcounter{subsection}{0}
       \renewcommand{\thesubsection}{A\arabic{section}.\arabic{subsection}}
       \setcounter{equation}{0}
       \renewcommand{\theequation}{A\arabic{equation}}
    \makeatother}

\begin{document}
\maketitle

\begin{abstract}
  Individualized treatment rules, cornerstones of precision medicine, inform
  patient treatment decisions with the goal of optimizing patient outcomes.
  These rules are generally unknown functions of patients' pre-treatment
  covariates, meaning they must be estimated from clinical or observational
  study data. Myriad methods have been developed to learn these rules, and these
  procedures are demonstrably successful in traditional asymptotic settings with
  moderate number of covariates. The finite-sample performance of these methods
  in high-dimensional covariate settings, which are increasingly the norm in
  modern clinical trials, has not been well characterized, however. We perform a
  comprehensive comparison of state-of-the-art individualized treatment rule
  estimators, assessing performance on the basis of the estimators' rule
  quality, interpretability, and computational efficiency. Sixteen
  data-generating processes with continuous outcomes and binary treatment
  assignments are considered, reflecting a diversity of randomized and
  observational studies. We summarize our findings and provide succinct advice
  to practitioners needing to estimate individualized treatment rules in high
  dimensions. Owing to these estimators' poor interpretability, we propose a
  novel pre-treatment covariate filtering procedure based on recent work for
  uncovering treatment effect modifiers. We show that it improves estimators'
  rule quality and interpretability. All code is made publicly available,
  facilitating modifications and extensions to our simulation study.
\end{abstract}

\keywords{clinical trials \and heterogeneous treatment effects \and
  observational studies \and precision medicine}

\section{Introduction}
\label{sec:intro}

Broadly, clinical trials are performed to evaluate the efficacy and safety of
novel therapies relative to the standard of care. Efficacy is generally assessed
by performing inference on a pre-defined marginal parameter, like the average
treatment effect. The rise of personalized medicine is however partially at odds
with this approach, especially when the patient population is large and diverse.
Seemingly similar patients' response to a given treatment can vary widely, which
is often attributable to the existence of unknown patient subpopulations.
Identifying these patient subgroups is therefore required to maximize patient
outcomes and optimize trial success rates.

These subpopulations are generally defined as an unknown function of patient
characteristics, like age, sex-at-birth, or genetic mutations. Certain
pre-treatment covariates are therefore said to modify the effect of a treatment.
These variables are referred to as \textit{treatment effect modifiers} (TEMs).

Individualized treatment rules (ITRs) may in turn be defined using these TEMs.
Much effort has been dedicated to developing statistical methods for learning
ITRs, and in particular \textit{optimal} ITRs \citep[see ][ to name but a
few]{robins2004, qian2011, zhang2012a, zhang2012b, tian2014, luedtke2016,
  kunzel2019}. Here, an optimal rule refers to that which optimizes the average
clinical outcome in the patient population. Note too that ITRs may be inferred
from clinical trial and observational study data alike.

Beyond informing patient treatment decisions, estimated ITRs may provide some
insight on the biological functioning of therapies by classifying pre-treatment
covariates as TEMs. TEM classification generally requires that the rules be
inferred using ``interpretable'' methods \citep[see ][ for a discussion on
interpretability in machine learning]{murdoch2019}. Examples include penalized
linear modeling procedures and decision trees, methods capable of communicating
which covariates are believed to interact with the treatment in order to predict
outcomes. Importantly, rules estimated by interpretable ITR estimators can be
vetted by domain experts, possibly increasing the rules trustworthiness from the
perspective of patients and healthcare providers alike.

While many approaches can successfully estimate the optimal ITR in randomized
settings where the number of pre-treatment covariates is small relative to the
number of patients, it is not necessarily the case when there are more
pre-treatment covariates than observations. The challenge is even greater when
estimating this rule from observational study data due to possible confounding.
As in many high-dimensional settings, simplifying assumptions about the
complexity of the estimand must be made to ensure consistent estimation and
therefore reliable predictive performance. Examples include sparsity, such as
subpopulation-specific expression of a particular biomarker
\citep{qian2011,tian2014,chen2017,zhao2021,bahamyirou2022}, or that the outcome
is a (partially) linear function of the covariates and treatment
\citep{tian2014,chen2017}. To the best of our knowledge, there neither exists a
comprehensive comparison of popular ITR estimators in high-dimensional settings
or an analysis of these methods' sensitivity to violations of their assumptions.

Recent developments for uncovering TEMs in high-dimensional data have been
proposed by \citet{boileau2022a,boileau2023} that might partially alleviate the
difficulty of estimating optimal ITRs in these settings. These authors define a
formal statistical framework for treatment effect modifier variable importance
parameters (TEM-VIPs), permitting the reliable classification of individual
pre-treatment covariates as TEMs. The TEM classification methods of
\citet{boileau2022a,boileau2023} control the false positive rate at the nominal
level while maintaining elevated power in high-dimensional data-generating
processes (DGPs).

In the same way that variables can be filtered prior to fitting traditional
regression methods to improve interpretability and predictive performance, we
posit that these TEM-VIP inference procedures may be used with any of the
existing ITR estimators to a similar effect. We refer to this procedure as
\textit{TEM-VIP-based filtering}. Pre-treatment covariates are classified as
TEMs or non-effect modifiers. Inferred TEMs --- in addition to all other
variables known to be clinically relevant --- are then be used by an ITR
estimator to learn the treatment rule. This filtering strategy may be
particularly useful when combined with black-box ITR methods, which generally
provide greater predictive performance than more rigid parametric methods at the
expense of opaqueness.

\reviewer{Our contributions to the ITR estimation literature are hence two-fold:
  \begin{enumerate}
    \item Provide practical guidance to applied scientists on ITR estimation in
          high-dimensional settings, and
    \item Assess the impact of TEM-VIP-based filtering on ITR estimation.
  \end{enumerate}
  A comprehensive simulation study using a variety of DGPs representative of
  randomized controlled trials (RCTs) and observational studies with continuous
  outcomes, binary treatment variables, and high-dimensional covariates is
  performed to accomplish both tasks. As part of this simulation study, ITR
  estimators, with and without TEM-VIP-based filtering, are compared on the
  basis of rule quality, interpretability, and computational efficiency.}

\section{Problem Formulation}
\label{sec:prob-form}

We begin by formally outlining the statistical problem. Consider $n$ independent
and identically distributed (i.i.d.) random vectors
$X_i = (W_i, A_i, Y_i^{(0)}, Y_i^{(1)}) \sim P_{X,0} \in \mathcal{M}_X$ for,
$i = 1, \ldots, n$, representing the full, partially-unobserved data. Indices
are omitted for convenience where possible throughout the remainder of the text.
In the full-data random vector, $X$, $W$ is a vector of $p$ pre-treatment
covariates, $A$ is a binary indicator of treatment assignment, and $Y^{(0)}$ and
$Y^{(1)}$ are random variables representing observations' potential outcomes
under assignment to the control and treatment groups, respectively
\citep{rubin1974}. We assume that $p$ is approximately equal to or larger than
$n$. Further, $Y^{(0)}$ and $Y^{(1)}$ are either continuous or binary and it is
assumed that large values of $Y$ are beneficial; this has no bearing on the
definition of the estimand, the optimal ITR. The true full-data DGP $P_{X,0}$ is
typically unknown and is a member of the nonparametric statistical model
$\mathcal{M}_X$, the collection of all possible full-data DGPs.

Only one of the potential outcomes is generally observed for any given
participant; the full data are censored by the treatment assignment mechanism.
However, this characterization of the full-data random vectors permits a
rigorous definition of the optimal ITR and related quantities.

In particular, an intermediate parameter for performing inference about the
optimal ITR is the conditional average treatment effect (CATE):
\begin{equation}\label{eq:CATE}
  \Gamma_{X,0}(W) \equiv
  \mathbb{E}_{P_{X,0}}\left[Y^{(1)}-Y^{(0)}\big| W\right] \;.
\end{equation}
This functional corresponds to the expected difference of individuals' potential
outcomes conditional on pre-treatment covariates. For any given participant with
$W=w$, a $\Gamma_{X,0}(w)$ larger than zero indicates that, on average,
participants with the same pre-treatment covariates benefit more from being
assigned to the treatment group than to the control group. The opposite is true
when $\Gamma_{X,0}(w)$ is less than zero.

With the CATE in hand, we define the optimal ITR as follows:
\begin{equation}\label{eq:opt-ITR}
  \Psi_{X,0}(W) \equiv I(\Gamma_{X,0}(W) > 0) \; .
\end{equation}
That is, the optimal ITR is an indicator function that equals one when the CATE
is larger than zero and zero otherwise. Proposition~\ref{prop:optimal-itr-mean}
follows:

\begin{prop}\label{prop:optimal-itr-mean}
  Let $\mathbf{\Psi}_{X}$ be the collection of all possible ITRs for the
  full-data model, i.e., the set of all functions from the covariate space into
  $\{0, 1\}$. Then
  \begin{equation*}
    \Psi_{X,0} = \argmax_{\Psi \in \mathbf{\Psi}_X} \;
    \mathbb{E}_{P_{X,0}}\left[
      Y^{(\Psi(W))}
    \right] \;.
  \end{equation*}
\end{prop}
\begin{proof}
  \begin{align*}
    \Psi_{X,0}
    & = \argmax_{\Psi \in \mathbf{\Psi}_X}\;\mathbb{E}_{P_{X,0}}\left[Y^{(\Psi(W))}\right] \\
    & = \argmax_{\Psi \in \mathbf{\Psi}_X}\;\mathbb{E}_{P_{X,0}}\left[\mathbb{E}_{\
      P_{X,0}}\left[Y^{(1)}\Psi(W) + Y^{(0)}(1 - \Psi(W))\big | W\right]\right]\\
    & = \argmax_{\Psi \in \mathbf{\Psi}_X} \;\mathbb{E}_{P_{X,0}}\left[\Psi(W)\mathbb{E}_{P_{X,0}}\left[Y^{(1)}-Y^{(0)}\big | W\right]\right]\\
    & = \argmax_{\Psi \in \mathbf{\Psi}_X} \;\mathbb{E}_{P_{X,0}}\left[\Psi(W) \Gamma_{X,0}(W) \right].
  \end{align*}    
\end{proof}
$\Psi_{X,0}$ is therefore optimal in the sense that it maximizes the expected
outcome under an ITR from among all possible ITRs. Note too that it is not
necessarily unique.

Now, as previously mentioned, only one of the potential outcomes is observed for
a given random unit. In place of $\{X_i\}_{i=1}^n$, we observe $n$ i.i.d. random
vectors $O = (W, A, Y) \sim P_0 \in \mathcal{M}$, where $W$ and $A$ are defined
as above, and $Y = AY^{(1)} + (1-A)Y^{(0)}$. $P_0$ is defined as the unknown
observed-data DGP, and is fully determined by $P_{X,0}$. This treatment
assignment mechanism is denoted by $\pi_{0}(W) \equiv \mathbb{P}_{P_{0}}[A=1|W]$
throughout the text. Note that this treatment assignment mechanism is known in
RCTs and is unknown in observational studies. Additionally, we refer to the
conditional expected outcome $\mathbb{E}_{P_{0}}[Y|A,W]$ by $\mu_{0}(W,A)$.

Regrettably, performing inference about the parameters of
Equations~\eqref{eq:CATE} and \eqref{eq:opt-ITR} using the observed data in
place of the full data is generally not possible without making additional
assumptions. Identifiability conditions linking parameters of the observed data
DGP to the desired parameters of the full data DGP are provided in
Proposition~\ref{prop:id-cont-dgp}.

\begin{prop}\label{prop:id-cont-dgp}
  Under the assumptions of no unmeasured confounding between the treatment and
  the outcome conditional on $W$, i.e., $A \perp Y^{(a)} \mid W$ for
  $a \in \{0, 1\}$, and that the probability of receiving treatment conditional
  on $W$ is bounded away from zero, i.e., $\pi_{0}(W) > 0$:
  \begin{equation*}
    \begin{split}
      \Gamma_0(W)
      & \equiv \mu_{0}(W, 1) - \mu_{0}(W, 0) \\
      & = \Gamma_{X,0}(W) \;,
    \end{split}
  \end{equation*}
  and
  \begin{equation*}
    \begin{split}
      \Psi_0(W)
      & \equiv I(\Gamma_0(W)>0) \\
      & = \Psi_{X,0}(W) \;.
    \end{split}
  \end{equation*}
\end{prop}
\begin{proof}
  This result follows immediately from the fact that
  $\mu_{0}(W,a) = \mathbb{E}_{P_{X,0}}[Y^{(a)}|W]$ for $a \in \{0,1\}$ under the
  stated assumptions.
\end{proof}

Proposition~\ref{prop:id-cont-dgp} establishes that the optimal ITR is estimable
from the observed data in the continuous and binary outcome setting. Its
assumptions are satisfied in RCTs: randomization ensures that (1) there are no
treatment--outcome confounders, (2) the treatment assignment mechanism is known,
and (3) the probability of receiving treatment, conditional on covariates, is
generally bounded away from zero. In \reviewer{observational studies accounting
  for confounders}, these assumptions are such that the data may be considered
as if generated by a \reviewer{stratified} RCT but with an unknown treatment
assignment mechanism.

\reviewer{We note that, when the confounding variables and treatment effect
  modifiers are known a priori, they can be used to define the CATE and the
  optimal ITR instead of the entire pre-treatment covariates vector $W$. It is
  straightforward to show that a result analogous to
  Proposition~\ref{prop:id-cont-dgp} can be obtained using this (possibly)
  reduced collection of variables. When identifiability conditions are
  satisfied, the CATE only depends on the set of treatment effect modifiers. In
  practice, however, the relevant subset of pre-treatment covariates is
  generally unknown. Using numerous pre-treatment covariates is therefore
  encouraged to ensure that Proposition~\ref{prop:id-cont-dgp} is satisfied and
  that many treatment effect modifiers are considered.}

Assuming Proposition~\ref{prop:id-cont-dgp} is satisfied, we take as target of
inference the optimal ITR, $\Psi_{0}(W)$, which, as a preliminary step, requires
the estimation of the CATE, $\Gamma_{0}(W)$. Possible estimators are presented
and discussed in Section~\ref{sec:methods}.

\section{Methods}
\label{sec:methods}

Myriad methods have been developed for estimating the CATE and therefore the
optimal ITR. We review a subset of these approaches capable of performing
inference about the latter in the high-dimensional, binary treatment setting. We
begin with the simplest parametric approach and end with state-of-the-art
nonparametric methods. Throughout, we represent estimators of
$\Gamma_{0}(W), \Psi_{0}(W), \mu_{0}(W,A)$, and $\pi_{0}(W)$ fit using the
empirical distribution $P_{n}$ by $\Gamma_{n}(W), \Psi_{n}(W), \mu_{n}(W,A)$,
and $\pi_{n}(W)$, respectively.

\subsection{Plug-In Estimators}
\label{subsec:parametric-methods}

The most straightforward approach for estimating $\Gamma_0$ is to assume that
the conditional expected outcome admits a parametric form, like that of a
generalized linear model (GLM). For example, we might assume that
$\mu_0(W, A) = l^{-1}(\alpha_{0} + \beta_{0} A + \gamma_{0}^\top W + \delta_{0}^\top WA)$
for some link function $l(\cdot)$. \reviewer{Here, $\alpha_{0}$ and $\beta_{0}$
  are intercept terms, and $\gamma_0$ and $\delta_{0}$ are $p$-length vectors.}
Letting the outcome be continuous and using the identity link, $l(x) = x$,
$\mu_{0}(W,A)$ could be estimated using penalized regression methods, like the
LASSO \citep{tibshirani1996} or the elastic net \citep{zou2005}, to produce
coefficient estimates $\alpha_n, \beta_{n}, \gamma_{n},$ and $ \delta_{n}$.
\reviewer{Note that, when the pre-treatment covariates are centered such that
  $\mathbb{E}_{P_{0}}[W] = 0$, $\alpha_{0} = \mathbb{E}_{P_{X,0}}[Y^{(0)}]$ and
  $\beta_{0} = \mathbb{E}_{P_{X,0}}[Y^{(1)}-Y^{(0)}]$ --- the average treatment
  effect --- under the identity link.}

The plug-in CATE estimator under this linear model is given by

\begin{equation}\label{eq:lm-cate}
  \begin{split}
    \Gamma_n(W)
    & \equiv \mu_{n}(W, 1) - \mu_{n}(W, 0) \\
    & = \beta_n + \delta_{n}^\top W \;.
  \end{split}
\end{equation}
This plug-in estimator can then be used to construct an optimal ITR estimator
under this linear model:
\begin{equation}\label{eq:lm-itr}
  \Psi_n(W) \equiv \text{I}(\Gamma_n(W) > 0).
\end{equation}
Analogous estimators can be constructed for scenarios with binary outcomes using
penalized logistic regression to estimate the conditional expected outcome.

While computationally efficient implementations can be painlessly constructed in
standard software, like the \texttt{R} language and environment for statistical
computing \citep{rstats}, and resulting estimators are generally interpretable,
this approach may be hampered by its reliance on strong parametric assumptions
about the functional form of the CATE. When the parametric form of the
conditional expected outcome is misspecified, the estimator of
Equation~\eqref{eq:lm-cate} could be biased, and so too could be the
corresponding ITR estimator of Equation~\eqref{eq:lm-itr}. Further, when the
covariates are correlated, the covariates classified as TEMs by penalized
regression methods like the LASSO---variables whose corresponding treatment
interaction coefficient estimates in $\delta_{n}$ are non-zero---may be
unreliable \citep{zhao2006}.

We might instead estimate $\mu_0(W, A)$ using more flexible machine learning
procedures, like Random Forests \citep{breiman2001} or XGBoost \citep{chen2016}.
In fact, $\mu_{0}(W,0)$ and $\mu_{0}(W,1)$ can even be estimated individually on
observations in control and treatment conditions, respectively, using possibly
different estimators. Estimators of the CATE and the accompanying optimal ITR
are then constructed similarly to those of Equations~\eqref{eq:lm-cate} and
\eqref{eq:lm-itr}.

These flexible plug-in estimators share similar limitations to their parametric
counterparts, however: if the estimators of $\mu_{0}(W,0)$ and $\mu_{0}(W,1)$
are not consistent, then neither will their associated CATE estimator. If
estimated using black-box methods, then they may also be less interpretable and
computationally efficient than plug-in estimators obtained from parametric
modeling procedures.

\subsection{(Augmented) Modified Covariates Estimators}
\label{subsec:modified-cov}

In part motivated by the limitations of the parametric plug-in estimator,
\citep{tian2014} developed a semiparametric framework for estimating the CATE in
RCTs. Assume for now that the outcome is a continuous variable and let
\begin{equation}\label{eq:mod-out-model}
  \mu_0(W, A) = \gamma_{0}^\top f(W) + \delta_{0}^\top h(W) \frac{2A-1}{2} \;,
\end{equation}
where $f$ and $h$ are some functions of the covariates. For simplicity of
presentation, we let $f(W) = h(W) = W$ \reviewer{such that $\gamma_{0}$ and
  $\delta_{0}$ are $p$-length coefficient vectors}. It follows that
$\Gamma_{0}(W) = \mathbb{E}_{P_0}[2Y(2A-1) | W] = \delta_{0}^\top W$. Dubbed the
\textit{modified outcome} method, this approach permits a straightforward
estimation of the CATE using standard ordinary least squares or penalized
methods like the LASSO and elastic net. As \citep{tian2014} note, however, this
framework does not easily generalize to DGPs with other kinds of outcome
variables.

They instead propose the more general \textit{modified covariates} method, which
relies on GLMs or the proportional hazards model. Again, first assuming a
continuous outcome and positing the following working model,
\begin{equation}\label{eq:mod-cov-model-cont}
  \mu_{0}(W,A) = \delta_{0}^\top h(W) \frac{2A-1}{2} \;,
\end{equation}
\citet{tian2014} show that the $\delta_{0}$s of
Equations~\eqref{eq:mod-out-model} and~\eqref{eq:mod-cov-model-cont} share
identical interpretations and can be estimated in the same manner. We again let
$h(W) = W$. The resulting CATE estimator in this model therefore takes the form
of
\begin{equation*}
  \Gamma_n(W) \equiv \delta_{n}^\top W \;,
\end{equation*}
where $\delta_{n}$ is again estimated via (penalized) linear methods. The
accompanying optimal ITR estimator take the same form as the estimator of
Equation~\eqref{eq:lm-itr}.

GLMs can be used to apply this approach to other kinds of outcomes. When $Y$ is
binary, the conditional outcome regression can be represented by the following
logistic regression model:
\begin{equation}\label{eq:mod-cov-model-bin}
  \mu_0(W,A) = \frac{\exp(\delta_{0}^\top h(W) (2A-1)/2)}
   {1 + \exp(\delta_{0}^\top h(W) (2A-1)/2)} \;.
\end{equation}
When this working model is well-specified,
\begin{equation*}
  \Gamma_0(W)=\frac{\exp(\delta_{0}^\top h(W)/2)-1}
  {\exp(\delta_{0}^\top h(W)/2)+1}.
\end{equation*}
Estimates of $\delta_{0}$ are again obtained by fitting the working model for
the conditional expected outcome using standard (penalized) regression
approaches. An optimal ITR estimator under this model is constructed as in
Equation~\eqref{eq:lm-itr}.

While the modified covariates framework is more flexible than the parametric
approaches based on GLMs, it is generally inefficient. \citet{tian2014}
therefore proposed an ``augmented'' modified covariates framework that can be
used to construct CATE estimators with generally smaller variance but which are
asymptotically equivalent. In the continuous outcome scenario with $h(W) = W$,
the resulting CATE estimator is equivalent to $\Gamma_n(W)$ of
Equation~\eqref{eq:lm-cate}. This simplicity is not shared with non-linear
working models like that of Equation~\eqref{eq:mod-cov-model-bin}, however
\citep{tian2014}.

\citet{chen2017} generalized the (augmented) modified covariates methodology
through a loss-based estimation procedure relying on loss functions
incorporating the propensity score. For a continuous outcome, the non-augmented
CATE estimator is defined as
\begin{equation*}
  \Gamma_{n}(W) \equiv
  \argmin_{\Gamma \in \boldsymbol{\Gamma}_{n}} \;
  \sum_{i=1}^{n}
  \frac{(Y_{i} - (2A_{i}-1)\Gamma(W_{i}))^{2}}
    {(2A_{i}-1)\pi_{0}(W_{i}) + 1-A_{i}} \;,
\end{equation*}
where $\boldsymbol{\Gamma}_{n}$ is the set of possible CATE estimators and
$\pi_{0}(W)$ must be replaced by $\pi_{n}(W)$ in observational studies. This
generalized framework also permits more flexible estimation of $\mu_{0}(W,A)$
using machine learning methods like with XGBoost. More efficient estimation of
the CATE and the ITR is therefore made possible using this modified covariates
model in observational studies, assuming that $\mu_{n}(W,A)$ and $\pi_{n}(W)$
are consistent. Readers interested in the augmented version of this weighted
CATE estimator are invited to review \citet{chen2017}.

Like the parametric plug-in estimators, these approaches result in biased CATE
estimators when their models are misspecified. This may result in a biased ITR
estimator, too. When the difference in expected outcomes is estimated using more
complex methods, like XGBoost, or penalized GLMs are used but covariates are
correlated \citep{zhao2006}, the interpretability of the resulting ITR estimate
is reduced, too.

A computationally efficient implementations of this methodology is made
available in the \texttt{personalized} \texttt{R} package \citep{huling2021}. We
use it throughout the simulation studies in Sections~\ref{sec:simulations}.

\subsection{Nonparametric Approaches}
\label{subsec:np-methods}

Much work in the nonparametric estimation literature has been dedicated to
developing estimators of the CATE and the optimal ITR; see \citet{robins2004,
  robins2008, zhang2012a, zhang2012b, vdl2015, luedtke2016, zhao2021,
  bahamyirou2022, kennedy2022optimal}. These estimators, typically relying on
nuisance parameters like $\mu_{0}(W,A)$ and $\pi_{0}(W)$, are generally
doubly-robust. That is, these ITR estimators are consistent estimators under
mild regularity conditions and so long as either $\mu_{n}(W,A)$ or $\pi_{n}(W)$
is consistent. Given that the treatment assignment rule $\pi_{0}(W)$ is
typically known in a clinical trial, these nonparametric estimators are
guaranteed to be consistent in these settings.

\reviewer{\subsubsection{Estimators Based on the Augmented Inverse Probability
    Weighted Transform}}

We briefly present the approach inspired by \citet{luedtke2016}. It relies on
the Augmented Inverse Probability Weighted (AIPW) transform, defined as
\begin{equation*}
    T(O; \mu_{n}, \pi_{n}) = \frac{2A - 1}{A\pi_{n}(W) + (1-A)(1 - \pi_{n}(W))}
      (Y-\mu_{n}(W,A)) + \mu_{n}(W,1) - \mu_{n}(W,0)\;,
\end{equation*}
and also referred to as the \textit{pseudo-outcome difference}.

A nonparametric, ``AIPW-based'' ITR estimator can then be defined as follows:
\begin{equation*}
  \Gamma_{n}(W) \equiv \argmin_{\Gamma \in \boldsymbol{\Gamma}_{n}}\;
  \sum_{i=1}^{n}\left(\Gamma(W_i) - T(O_i; \mu_{n}, \pi_{n})\right)^{2} \;.
\end{equation*}
That is, $\Gamma_n(W)$ is the squared AIPW-transform error risk minimizer from
among the set of possible CATE estimators. $\Gamma_n(W)$ can be constructed by
regressing the pseudo-outcome differences given by $T(O;\mu_{n},\pi_{n})$ on the
covariates $W$. This estimator is a consistent estimator of $\Gamma_0(W)$ when
either $\mu_{n}(W,A)$ or $\pi_{n}(W)$ are consistent \citep{luedtke2016}. The
accompanying optimal ITR estimator is constructed similarly to the ITR
estimators of Equation~\eqref{eq:lm-itr}.

This nonparametric approach makes minimal assumptions about the DGP and is made
all the more attractive by the possible use of flexible machine learning
algorithms to both estimate nuisance parameters and the CATE, curbing the risk
of model misspecification and potentially translating to increased finite sample
precision. We recommend using an estimator based on the Super Learner framework
of \citet{laan2007}. Building upon the theory of cross-validated loss-based
estimation \citep{laan-dudoit2003a, laan-dudoit2003b, dudoit2005asymptotics,
  vaart2006}, a Super Learner constructs a convex combination of estimators from
a pre-specified library that minimizes the cross-validated risk of a pre-defined
loss. This collection of estimators can be made up of modern machine learning
algorithms, avoiding the need to make strong parametric assumptions about the
target parameter. The resulting estimator converges in probability to the oracle
asymptotically. The oracle corresponds to the estimator that would be selected
for the given dataset if $P_{0}$ were in fact known. The Super Learner is
conveniently implemented in the \texttt{sl3} R package \citep{coyle2021sl3},
which we use in the subsequent simulation study.

This flexibility of this estimation procedure comes at the cost of computational
efficiency, however. Whether the aforementioned asymptotic properties lead to
improved performance in the high-dimensional settings considered here remains
uncertain, too.

The resulting CATE estimator's level of interpretability also largely depends on
the method used to regress pseudo-outcome differences on the covariates. A
compromise between flexibility and finite sample precision might be achieved by
fitting the pseudo-outcome differences as a function of the covariates using
penalized GLMs \citep{zhao2021,bahamyirou2022}. If interpretability is not a
concern, then a Super Learner composed of flexible machine learning algorithms
will likely provide the best finite sample performance.

\reviewer{\subsubsection{Causal Random Forests}}

An alternative but related nonparametric approach is that of Causal Random
Forests \citep{wager2018,athey2019}. Like the Random Forests algorithm proposed
by \citet{breiman2001}, a large number of decision trees are constructed on
random subsamples of the data. Whereas traditional Random Forests generate
branches in any given tree by maximizing the variance of the outcome within
subgroups formed across random selections of covariates, the causal version
attempts to maximize the variance of the estimated treatment effects. The CATE
is then constructed using a weighted Robinson's residual-on-residual regression
\citep{robinson1988}. This methodology is implemented in the \texttt{grf}
\texttt{R} package \citep{grf}, which we use in the simulation study of
Section~\ref{sec:simulations}.

Similar to the previously described nonparametric CATE estimator, $\mu_{n}(W,A)$
and $\pi_{n}(W)$ must be fit. This is typically done using traditional Random
Forests. If either estimator is consistent, then so too is the CATE estimator
given by the Causal Random Forests under mild conditions on the DGP. Reliably
recovering TEMs is however not possible in high dimensions.

\subsection{Feature Selection using Treatment Effect Modifier Variable Importance
  Parameters}
\label{subsec:feature-selection}

Recent work by \citet{boileau2022a,boileau2023} discussed and demonstrated how
CATE-based methods aiming to uncover TEMs in high-dimensional RCTs and
observational studies generally fail to reliably uncover true TEMs. This has
been reported elsewhere as well \citep[for example, see the simulation studies
of ][]{tian2014,bahamyirou2022}. This suggests that the previously reviewed CATE
estimators may needlessly incorporate uninformative covariates in their
estimation procedures, potentially decreasing their rule quality and diminishing
their interpretability.

Motivated by the need for reliable TEM discovery methodology applicable to
high-dimensional data, \citet{boileau2022a, boileau2023} proposed a
nonparametric framework for defining TEM-VIPs based on any pathwise
differentiable treatment effect and for subsequently constructing regular
asymptotically linear (RAL) estimators of these parameters. This framework
provides formal statistical guarantees, like Type I error rate control.

Of particular relevance is the absolute TEM-VIP based on the average treatment
effect (ATE), and therefore related to the CATE. Assuming that the expected
value of $\mu_{0}(W,1)-\mu_{0}(W,0)$ conditional on any given $W_{j}$ is linear
in $W_{j}$, this TEM-VIP is defined as the simple linear regression coefficient
obtained by regressing the difference in conditional expected potential outcomes
on a single covariate. For the $j^{\text{th}}$ pre-treatment covariate $W_{j}$,
this parameter is given by
\begin{equation*}
  \frac{\text{Cov}_{P_{0}}\left[
      \mu_{0}(W, 1) - \mu_{0}(W, 0), W_{j}
    \right]}{\text{Var}_{P_{0}}[W_{j}]} .
\end{equation*}
Though the true relationship between this expected outcome difference and the
covariate conditional on said covariate is unlikely to be linear, this parameter
is an informative measure of treatment effect modification: any value greater
than zero is indicative of effect modification, and the magnitude and sign of
the TEM-VIP summarize, respectively, the strength and direction of the effect.

\citet{boileau2022a, boileau2023} provide one-step, estimating equation, and
targeted maximum likelihood (TML) estimators of this TEM-VIP that possess
$\mu_{0}(W,A)$ and $\pi_{0}(W)$ as nuisance parameters, and prove that these
estimators are doubly-robust and efficient in nonparametric models. Assuming
$\mu_{n}(W,A)$ and $\pi_{n}(W)$ converge to $\mu_{0}(W,A)$ and $\pi_{0}(W)$,
respectively, at the appropriate nonparametric rate, $o_{P}(n^{-1/4})$, these
TEM-VIP estimators are shown to be RAL. In RCTs, these estimators are guaranteed
to be RAL by virtue of $\pi_{0}(W)$ being known. Recall that RAL estimators'
sampling distributions are asymptotically Normal, permitting computationally
efficient hypothesis testing. Inference procedures for this TEM-VIP are
implemented in the \texttt{unihtee} \texttt{R} package. We use this software in
the simulation studies.

We propose to use these TEM-VIP estimators in a two-stage CATE estimation
procedure. First, the TEM-VIPs are estimated for each covariate, and a
hypothesis test about their modification status is performed. Covariates
classified as TEMs based on a predefined null hypothesis and level of
significance are then used to estimate the CATE.

As previously mentioned, decreasing the number of covariates used by CATE
estimators in the second stage may decrease said estimators' variance and bias.
The two-stage procedure may also be more computationally efficient than fitting
an estimator using the entire set of covariates. Ranking the estimated TEM-VIPs
also adds a layer of interpretability to the resulting CATE estimate, regardless
of the CATE estimation strategy used in the second stage.

\section{Simulation Study}
\label{sec:simulations}

\subsection{Data-Generating Processes and Simulation Details}

Recall that, in a continuous outcome setting, we have access to realizations of
$O = (W, A, Y)$, where $W$ is a $p$-dimensional vector of pre-treatment
covariates (and possible confounders), $A$ is a binary treatment indicator, and
$Y$ is the observed continuous outcome. We consider generative models based on
the template below:
\begin{align*}
  W & \sim N(0, \Sigma) \\
  A|W & \sim \text{Bernoulli}(\pi(W)) \\
  Y^{(A)}|W & \sim N(\mu(W,A), 1) \\
  Y & = AY^{(1)} + (1-A)Y^{(0)} \;.
\end{align*}
Here, $\Sigma$ is some $p \times p$ covariance matrix and $\pi(W)$ and
$\mu(W,A)$ are generic propensity score and conditional expected outcome
functions, respectively.

Setting $p=500$, we define 16 DGPs using every possible combination of the
following factors:
\begin{equation*}
\resizebox{\textwidth}{!}{$%
  \begin{split}
    \Sigma_{1} & = I_{500 \times 500} \\
    \Sigma_{2} & = \text{Block diagonal} \\
  \end{split}
  \quad \times \quad
  \begin{split}
    \pi_{1}(W) & = \frac{1}{2} \\
    \pi_{2}(W) & = \text{logit}^{-1}\left(\frac{W_{1} + W_{2} + W_{3} + W_{4}}{5}\right) \\
  \end{split}
  \quad\times\quad
  \begin{split}
    \mu_{1}(A, W) & = A + \gamma^{\top}W + (\delta^{(10)})^{\top}WA \\
    \mu_{2}(A, W) & = A + \gamma^{\top}W + (\delta^{(50)})^{\top}WA \\
    \mu_{3}(A, W) & = \gamma^{\top}W + 2\;\text{arctan}\left\{(\delta^{(10)})^{\top}WA\right\} \\
    \mu_{4}(A, W) & = \gamma^{\top}W + 2\;\text{arctan}\left\{(\delta^{(50)})^{\top}WA\right\} \\
  \end{split}$%
  }
\end{equation*}
where $\gamma_{1} = \ldots = \gamma_{5} = 2$,
$\gamma_{6} = \ldots = \gamma_{500} = 0$,
$\delta^{(10)}_{1} = \ldots = \delta^{(10)}_{10} = 2$,
$\delta^{(10)}_{11} = \ldots = \delta^{(10)}_{500} = 0$,
$\delta^{(50)}_{1} = \ldots = \delta^{(50)}_{50} = 1/2$, and
$\delta^{(50)}_{51} = \ldots = \delta^{(50)}_{500} = 0$. $\Sigma_{2}$ is
constructed by randomly generating $50$ positive definite square symmetric
matrices with diagonal elements equal to one. Details on the generation of these
matrices are provided in the accompanying code.

We note that DGPs using $\pi_{1}$ are treated as observational studies in that
the treatment assignment mechanism is unknown. DGPs with $\pi_{2}$ mimic an RCT,
where $\pi_{2}$ is treated as known. We also highlight that $\delta^{(10)}$ and
$\delta^{(50)}$ permit the evaluation of methods in DGPs where the number of
TEMs is sparse and non-sparse, respectively. The effect sizes associated with
each TEM also vary in terms of the conditional expected outcome function:
$\mu_{1}$ produces TEMs with the largest effect sizes, followed by $\mu_{3}$,
$\mu_{2}$, and $\mu_{4}$ in decreasing order. Finally, DGPs using $\mu_{1}$ and
$\mu_{2}$ rely on linear models to produce the outcome, whereas DGPs using
$\mu_{3}$ and $\mu_{4}$ use non-linear models to generate outcomes as a
function of treatment assignment and pre-treatment covariates.

One hundred learning datasets made up of $n=250, 500$ and $1,\!000$ observations
are generated for each DGP. The potential outcomes $Y^{(0)}$ and $Y^{(1)}$ are
unobserved in these datasets. Accompanying each learning set is a test set made
up of $n^{\prime} = 100$ observations. These data contain the potential outcomes
for evaluation purposes. The learning datasets are used to fit the considered
CATE and ITR estimators which are subsequently assessed using the test sets.
\reviewer{The size of the learning and test datasets were chosen to resemble the
  sample sizes typically observed in Phase 2 and 3 oncology trials, as well
  those of observational studies investigating ITRs in moderate-to-small patient
  populations. Additionally, the considered learning set sample sizes produce
  datasets with a variety of $p/n$ ratios, offering insights into the
  benchmarked methods' performance across many high-dimensional regimes.}

\subsection{Estimators}

The CATE-based ITR estimators considered in this simulation study are listed in
Table~\ref{table:CATE-estimators}. \reviewer{Owing to the many possible
  implementations of the estimation approaches discussed in
  Section~\ref{sec:methods}, we restrict our benchmarks to ITR estimators
  relying on familiar nuisance parameter estimators capable of handling
  high-dimensional data that are easily implemented with standard statistical
  software. While an exhaustive benchmarking of ITR estimators is infeasible, we
  believe that this analysis provides valuable insights into ITR estimation in
  high dimensions. Additionally, this simulation study provides a foundation for
  bespoke investigations.}

Details on the ITR estimators' fitting procedures are provided in
Table~\ref{table:CATE-estimators}, as is a column indicating whether they
provide built-in TEM classification. \reviewer{That is, whether the CATE
  estimates underlying the ITRs explicitly distinguish between predicted TEMs
  and non-TEMs}. Among the considered estimators, only those relying solely on
the LASSO to estimate the conditional expected outcome possess this ability:
covariates with non-zero estimated treatment-covariate interaction coefficients
are categorized as TEMs\reviewer{, others as non-TEMs. We emphasize that while
  the LASSO-based ITR estimators provide built-in TEM classification, the
  resulting classifications may produce many false positives and negatives.
  Built-in TEM classification may be misleading and could therefore be
  considered an undesirable property in settings where the LASSO's assumptions
  are violated}. Note that all methods employing the LASSO select the
regularization hyperparameter by minimizing the 10-fold cross-validated mean
squared error.

\begin{table}[ht!]
  \resizebox{\textwidth}{!}{%
  \centering
  \begin{tabular}{
    |p{0.3\textwidth} || p{0.6\textwidth} | p{0.1\textwidth} |
    }
    \hline
    CATE Estimator
    & Details
    & Built-in TEM Classification \\
    \hline\hline
    Plug-In LASSO
    & A plug-in estimator, where $\mu(W,A)$ is estimated using the LASSO. Linear
      model coefficients for the treatment indicator, the pre-treatment covariates, and the
      treatment-covariates interactions are considered.
    & Yes \\
    Plug-In XGBoost
    & A plug-in estimator, where XGBoost \citep{chen2016} is used to estimate $\mu(W,A)$.
    & No \\
    Modified Covariates LASSO
    & A modified covariates estimator, where $\mu(W,A)$ is estimated with a LASSO linear regression. $\pi(W)$ is estimated in non-randomized DGPs with a logistic LASSO regression.
    & Yes \\
    Modified Covariates XGBoost
    & A modified covariates estimator, where $\mu(W,A)$ is estimated with XGBoost. $\pi(W)$ is
      estimated in non-randomized DGPs with a logistic LASSO regression.
    & No \\
    Augmented Modified Covariates LASSO
    & An augmented modified covariates estimator, where $\mu(W,A)$ is estimated with a LASSO linear regression. $\pi(W)$ is estimated in non-randomized DGPs with a logistic LASSO regression.
    & Yes \\
    Augmented Modified Covariates XGBoost
    & An augmented modified covariates estimator, where $\mu(W,A)$ is estimated with XGBoost. $\pi(W)$ is estimated in non-randomized DGPs with a logistic LASSO regression.
    & No \\
    AIPW-based LASSO
    & An AIPW-based estimator, where a Super Learner comprised of penalized linear regression (LASSO, ridge, elastic net), Random Forests, and XGBoost is used to estimate $\mu(W,A)$.
      A Super Learner comprised of penalized logistic regressions (LASSO, ridge,
      elastic net), Random Forests, and XGBoost is used to estimate $\pi(W)$
      when necessary. The difference in predicted pseudo-outcomes is regressed
      on the covariates using a LASSO linear regression.
    & Yes \\
    AIPW-based Super Learner
    & An AIPW-based estimator, where a Super Learner comprised of penalized linear regression (LASSO, ridge,
      elastic net), Random Forests, and XGBoost is used to estimate $\mu(W,A)$.
      A Super Learner comprised of penalized logistic regressions (LASSO, ridge,
      elastic net), Random Forests, and XGBoost is used to estimate $\pi(W)$
      when necessary. The difference in predicted pseudo-outcomes is regressed
      on the covariates using a Super Learner identical to that used to estimate
      the conditional expected outcomes.
    & No \\
    Causal Random Forests
    & A Causal Random Forests estimator, where Random Forests are used to estimate $\mu(W,A)$ and, when necessary, $\pi(W)$.
    & No \\
    \hline
  \end{tabular}}
  \caption{CATE-based ITR estimators benchmarked in simulation study.}
  \label{table:CATE-estimators}
\end{table}

A TEM-VIP-based filtered version of each estimator outlined in
Table~\ref{table:CATE-estimators} is also included in this benchmark. As
previously described in Section~\ref{sec:methods}, inference about TEM-VIPs is
performed to identify TEMs in each learning dataset replicate. These predicted
TEMs are then used to train the CATE-based ITR estimators.

We use the nonparametric doubly-robust one-step estimator to perform inference
about the absolute ATE-based TEM-VIPs \citep{boileau2022a,boileau2023}. Recall
that in observational studies, this filtering procedure requires that two
nuisance parameters be estimated: $\mu(W,A)$ and $\pi(W)$. They are fit in the
observational study DGPs considered here with the same Super Learners used by
the AIPW-based estimators listed in Table~\ref{table:CATE-estimators}. In the
RCT-like DGPs, $\pi(W)$ need not be estimated, and $\mu(W,A)$ is estimated with
a LASSO regression that includes main terms for the treatment and covariates, as
well as all possible treatment-covariate interaction terms. Again, regardless of
whether a DGP mimics an observational study or an RCT, this one-step estimator
is asymptotically Normal under mild conditions. Wald-type confidence intervals
are constructed for each pre-treatment covariate's TEM-VIP, and any variable
with an Benjamini-Hochberg FDR-adjusted $p$-value \citep{benjamini1995} lower
than $5\%$ is classified as a TEM.

\subsection{Estimator Performance Metrics}

Estimators are evaluated using three metrics: quality of the resulting ITR-based
classification on new data, accurate interpretability, and computational
efficiency.

\paragraph{Rule Quality:} For a generic ITR estimate $\psi_{b}(W)$ produced by
fitting a generic ITR estimator $\Psi_{n}(W)$ to the $b^{\text{th}}$ learning
dataset, we compute the mean outcome over the $b^{\text{th}}$ test set based on
the resulting estimate's classification as
\begin{equation*}
  \frac{1}{n^{\prime}} \sum_{i=1}^{n^{\prime}} Y_{i}^{(\psi_{b}(W_{i}))} \;,
\end{equation*}
where, recall, each test set is made up of $n^{\prime}=100$ observations.
$\Psi_{n}(W)$'s performance across the $B=100$ test set replicates is then
summarized as
\begin{equation*}
  \frac{1}{B} \sum_{b=1}^{B} \left(
    \frac{1}{n^{\prime}} \sum_{i=1}^{n^{\prime}} Y_{i}^{(\psi_{b}(W_{i}))}
  \right) \;.
\end{equation*}
This metric, corresponding to the empirical mean of the potential outcomes under
the ITR $\Psi_{n}(W)$, is therefore an estimator of
$\mathbb{E}_{P_{X,0}}[Y^{(\Psi_{X,0}(W))}]$, assuming the conditions of
Proposition~\ref{prop:id-cont-dgp} are satisfied. \reviewer{ITR estimators that
  produce the largest value are said to be ``high-quality''---their treatment
  assignments empirically optimize, relative to all considered ITRs, the
  observations' average outcomes.}

\paragraph{\reviewer{Accurate Interpretability}:} The non-zero entries in
$\delta^{(10)}$ and $\delta^{(50)}$ correspond to the TEMs of their respective
DGPs, and are therefore the driving force behind heterogeneous treatment
effects. For any given learning set replicate, the false discovery, true
negative, and true positive proportions with respect to TEM status are computed
for all estimators capable of classifying pre-treatment covariates. These
proportions are then averaged across replicates to produce the empirical false
discovery rate (FDR), true negative rate (TNR), and true positive rate (TPR).
ITR estimators are deemed accurately interpretable if the empirical false
discovery rate (FDR) approximately achieves the nominal Type I error rate of
$5\%$ while producing TNRs and TPRs near 100\%.

\paragraph{Computational Efficiency:} The time to fit each estimator, using a
single core, to each learning dataset is recorded and considered a surrogate for
overall computational efficiency. The mean fit time is then computed across
sample sizes and DGPs. Smaller times indicate improved computational efficiency.
We note that some of the estimators implementations permit parallelization,
which generally decreases the required time. We nevertheless fit all estimators
in serial to ensure a comparison that is relevant and informative for an
audience with access to a diverse set of computing environments.

\subsection{Results}

\reviewer{Due to the many DGPs and estimators considered, a representative
  subset of results is presented in Figure~\ref{fig:results-summary}. A
  comprehensive report of all estimators' performance with respect to the
  previously introduced metrics are provided in Tables~A1-A16 of the Appendix,
  stratified by DGP. These results are also illustrated in
  Figures~\ref{fig:expected-outcome}, \ref{fig:fdr}, \ref{fig:tpr},
  \ref{fig:tnr}, and \ref{fig:mean-fit-time}. We summarize our findings below.}

\begin{figure}
  \centering
  \includegraphics[width=\textwidth]{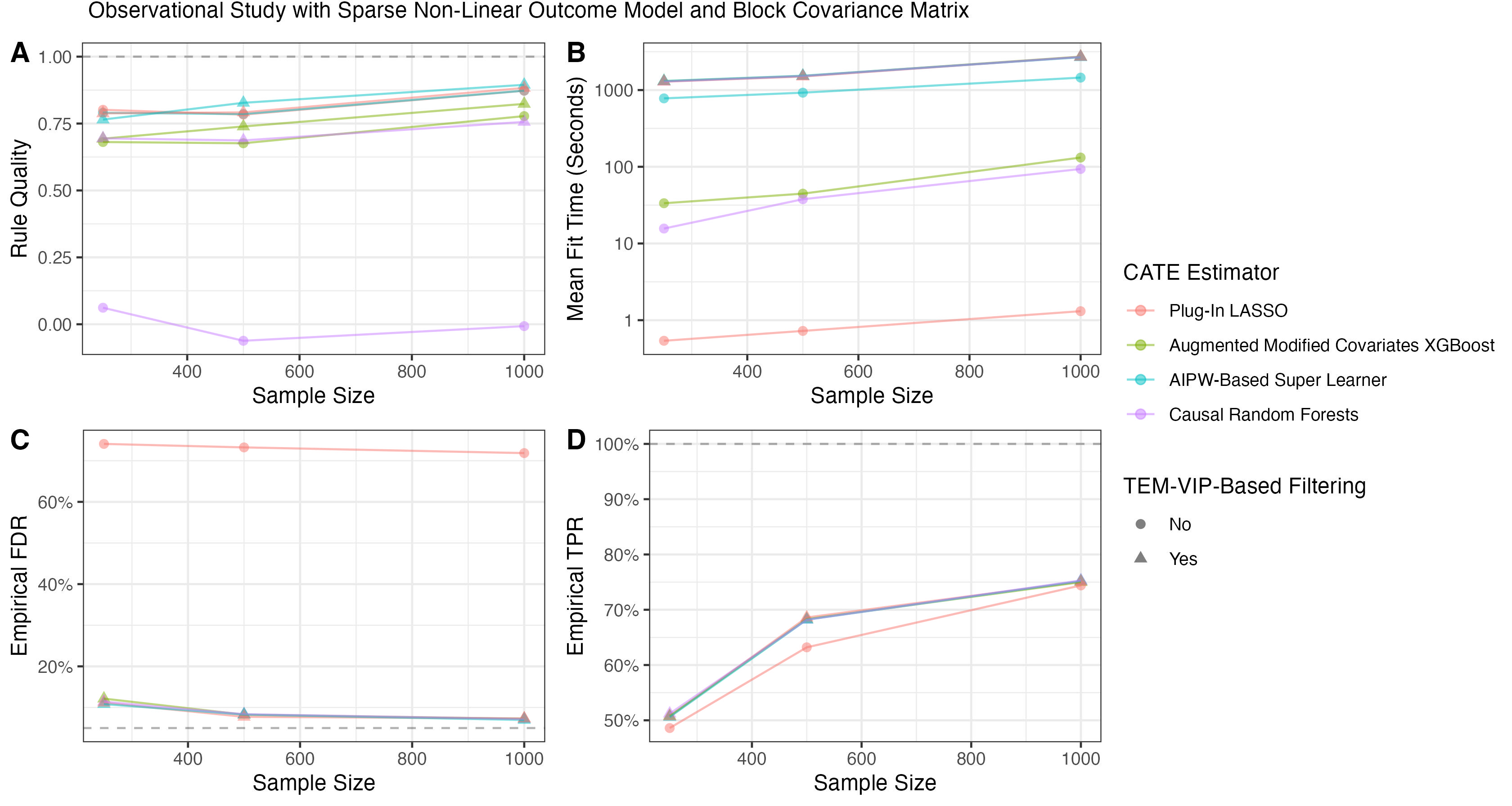}
  \caption{\textit{Simulation Study Results Summary}: The results of select
    estimators in the simulated observational study with sparse non-linear
    conditional expected outcome and a block covariance matrix are presented.
    These results are representative of the general findings of the simulation
    study. \textbf{(A)} The relative rule quality computed over the 100 test set
    replicates. Rule quality is defined as the mean ITR outcome divided by the
    optimal ITR, given in Proposition~\ref{prop:optimal-itr-mean}, which is
    approximated using a Monte Carlo procedure for each DGP. The dotted line
    corresponds to the idealized relative rule quality. \textbf{(B)} The mean
    fit time in seconds computed over the 100 learning set replicates. Note the
    y-axis's log scale. \textbf{(C)} The empirical FDR computed over the 100
    learning set replicates. The dotted line corresponds to desired nominal Type
    I error rate of $5\%$. \textbf{(D)}. The empirical TPR computed over the 100
    learning set replicates. The dotted line corresponds to the desired TPR of
    $100\%$.}
  \label{fig:results-summary}
\end{figure}

\paragraph{Rule Quality:} No estimator uniformly dominates the others
(Figure~\ref{fig:expected-outcome}). In the DGPs with sparse TEMs---regardless
of whether treatment is randomized or of the covariates' correlation
structure---the filtered and non-filtered versions of the LASSO-based plug-in,
AIPW-based, and augmented modified covariates estimators are the top performers
(Figure~\ref{fig:results-summary}A). These estimators produce empirical means
that are approximately equal to the population mean under the optimal ITR in the
linear model DGPs at all sample sizes and in the non-linear DGPs when
$n=1,\!000$. The optimal ITR is defined in
Proposition~\ref{prop:optimal-itr-mean}, which is approximated using a Monte
Carlo procedure for each DGP. When $n=250$, the non-filtered versions of these
estimators generally perform marginally better than their filtered counterparts.
The reverse is true as sample size increases.

In the non-sparse settings, the non-filtered plug-in LASSO, AIPW-based, and
LASSO-based augmented modified covariates estimators produce the best empirical
mean outcomes on average (Figures~\ref{fig:expected-outcome}). The empirical
means obtained under these estimators are near the population mean under the
optimal ITRs as sample size increases. The performance of the filtered versions
of these estimators' generally converges with increasing sample size to that of
their non-filtered counterparts. An exception to this trend is observed in the
DGPs with non-linear conditional expected outcomes and whose covariates'
covariance matrix is $I_{500\times 500}$.

The XGBoost-based plug-in, modified covariates, and Causal Random Forest
estimators are generally of lower rule quality than other estimators considered
here, though their performance is improved by the covariate filtering procedure.

\paragraph{Accurate Interpretability:} As illustrated in \citet{boileau2022a,
  boileau2023}, only ITR estimators employing the TEM-VIP-based filtering
procedure reliably control the FDR near the nominal $5\%$ level regardless of
the treatment assignment mechanism (Figures~\ref{fig:results-summary}C and
\ref{fig:fdr}). Indeed, these estimators produce empirical FDRs near
$\approx 5\%$ in all DGPs as sample sizes increase. All other estimators capable
of classifying TEMs produce empirical FDRs ranging from $40\%$ to $80\%$.
However, we find that this FDR control sometimes results in a worse empirical
TPR in non-sparse settings (Figure~\ref{fig:tpr}). This is not the case in DGPs
with sparse conditional expected outcome models
(Figure~\ref{fig:results-summary}D). With respect to the empirical TNR, only the
filtered estimators produce near-perfect results, and this regardless of sample
size (Figure~\ref{fig:tnr}).

\paragraph{Computational Efficiency:} In RCT DGPs, the non-filtered LASSO-based
plug-in estimator generally has the lowest mean fit time across all sample sizes
and DGPs, followed closely by the non-filtered (augmented) modified covariates
estimators using the LASSO and the non-filtered XGBoost plug-in estimator
(Figure~\ref{fig:mean-fit-time}). The non-filtered AIPW-based estimators are the
slowest estimators to produce a rule in all DGPs and sample sizes. All
estimators using the TEM-VIP-based filtering procedure exhibit similar mean fit
times that are several orders of magnitude larger than that of the non-filtered
LASSO-based plug-in estimator. These filtered estimators' mean fit times are
sometimes similar, however, to their non-filtered counterparts when $n=1,\!000$,
as is the case with Causal Random Forests and the augmented modified covariates
estimators using XGBoost. Further, the filtered AIPW-based estimators have lower
mean fit times than their non-filtered versions.

Similar patterns are observed in the observational study DGPs' simulations
(Figure~\ref{fig:results-summary}B). The non-filtered plug-in estimator using
the LASSO is again the most computationally efficient procedure, followed by the
non-filtered plug-in estimator using XGBoost (Figure~\ref{fig:mean-fit-time}).
The non-filtered (augmented) modified covariates estimators and non-filtered the
Causal Random Forests exhibit similar mean fit times which are approximately two
orders of magnitude larger than the fastest estimator. The nonparametric
AIPW-based estimators are again the slowest of the non-filtered estimators,
though they produce estimated rules more quickly than any of the estimators
relying on the TEM-VIP-based filtering strategy.

\section{Discussion}
\label{sec:discussion}

The simulation results suggest that no ITR estimator uniformly dominates all
others with respect to rule quality, \reviewer{accurate} interpretability, and
computational efficiency in the continuous outcome, binary treatment setting. We
posit that the same is true in binary outcome, binary treatment settings. A
trade-off must therefore be made when selecting a method with which to learn a
treatment rule. Given that the primary goal of learning an ITR is to optimize
patient benefit---that is, to learn the optimal ITR---we hold that only
high-quality estimators should be considered. Selecting a procedure therefore
comes down to favoring either \reviewer{accurate} interpretability or
computational efficiency as a secondary goal.

If more importance is placed on accurate interpretability than computational
efficiency, the TEM-VIP-based filtering procedure is a virtual necessity. None
of the non-filtered estimators considered here reliably control the empirical
FDR or produce near-perfect empirical TPRs (Figures~\ref{fig:fdr},
\ref{fig:tnr}). Even when empirical TPR optimization is required, filtered
estimators generally perform as well as non-filtered estimators
(Figure~\ref{fig:tpr}). In sparse TEM scenarios, high-qualiry estimators are
generally filtered. In non-sparse TEM settings, the filtered LASSO-based plug-in
and AIPW-based estimators are typically only marginally worse than the
non-filtered estimators with the highest rule quality.

However, the use of this filtering procedure may be constrained by available
computing power. In RCTs, this filtering procedure makes fitting ITR estimators
slower than their non-filtered versions in small sample sizes
(Figure~\ref{fig:mean-fit-time})---with the exception of nonparametric ITR
estimators. This gap decreases in larger sample sizes: the mean fit times for
the filtered estimators are approximately constant, while the non-filtered
estimators' fit times follow an approximately exponential function of sample
size. In observational studies, where more computationally intensive nuisance
parameter estimators are generally required to ensure accurate TEM
classification, the filtered estimators are generally several orders of
magnitude slower than their non-filtered counterparts.

When there is limited time, effort, or computing resources available to learn an
ITR, the simulation study results suggest that the non-filtered LASSO-based
plug-in estimator or the non-filtered augmented modified covariates approach
using XGBoost should be used for analyzing data collected in an RCT. In
observational studies, the LASSO-based plug-in estimator is the most
computationally efficient of the estimators with sufficiently elevated rule
quality. We again emphasize, however, that the computational efficiency of these
procedures comes at the cost of accurate interpretability: these estimators
cannot reliably distinguish TEMs from non-modifiers.

\section{Recommendations}
\label{sec:recommendations}

We summarize our discussion of the simulation results with concise
recommendations about the choice of ITR estimators based on desired operating
characteristics.

\paragraph{High-quality and accurately interpretable:} The filtered LASSO-based
plug-in and AIPW-based estimators generally produce high-quality ITR estimates
while accurately recovering TEMs. These estimators are computationally
intensive, however. When computational resources are available, parallelized
implementations of these estimators will decrease the time required to produce a
treatment rule estimate.

\paragraph{High-quality and computationally efficient:} The LASSO-based plug-in
estimator is among the most high-quality in our simulation study, providing
empirical evidence that it is robust to model misspecification while being
exceptionally computationally efficient. However, as previously mentioned, this
estimator's built-in feature selection capabilities should not be used for TEM
discovery.

We emphasize that while this simulation study is comprehensive, practitioners
may benefit from performing bespoke benchmarks of their own. Code used to
generate the simulations results of Section~\ref{sec:simulations}, found
on
\href{https://github.com/PhilBoileau/pub\_guidance-on-ITR-estimation-in-HD}{GitHub},
may serve as foundation for extensions and modifications. This code relies on
the \texttt{simChef} R package, which provides a general framework for building
reproducible and computationally efficient simulation studies
\citet{duncan2024}.

\FloatBarrier

\section*{Code}

Code used to generate the results is available in the
\href{https://github.com/PhilBoileau/pub\_guidance-on-ITR-estimation-in-HD}{PhilBoileau/pub\_guidance-on-ITR-estimation-in-HD}
GitHub repository.

\section*{Acknowledgments.}

P.B. gratefully acknowledges the support of the Fonds de recherche du Qu\'ebec -
Nature et technologies and the Natural Sciences and Engineering Research Council
of Canada.

\FloatBarrier

\bibliographystyle{abbrvnat}
\bibliography{references}

\newpage

\FloatBarrier

\section*{Appendix}
\startappendix

\begin{figure}
  \centering
  \includegraphics[width=\textwidth]{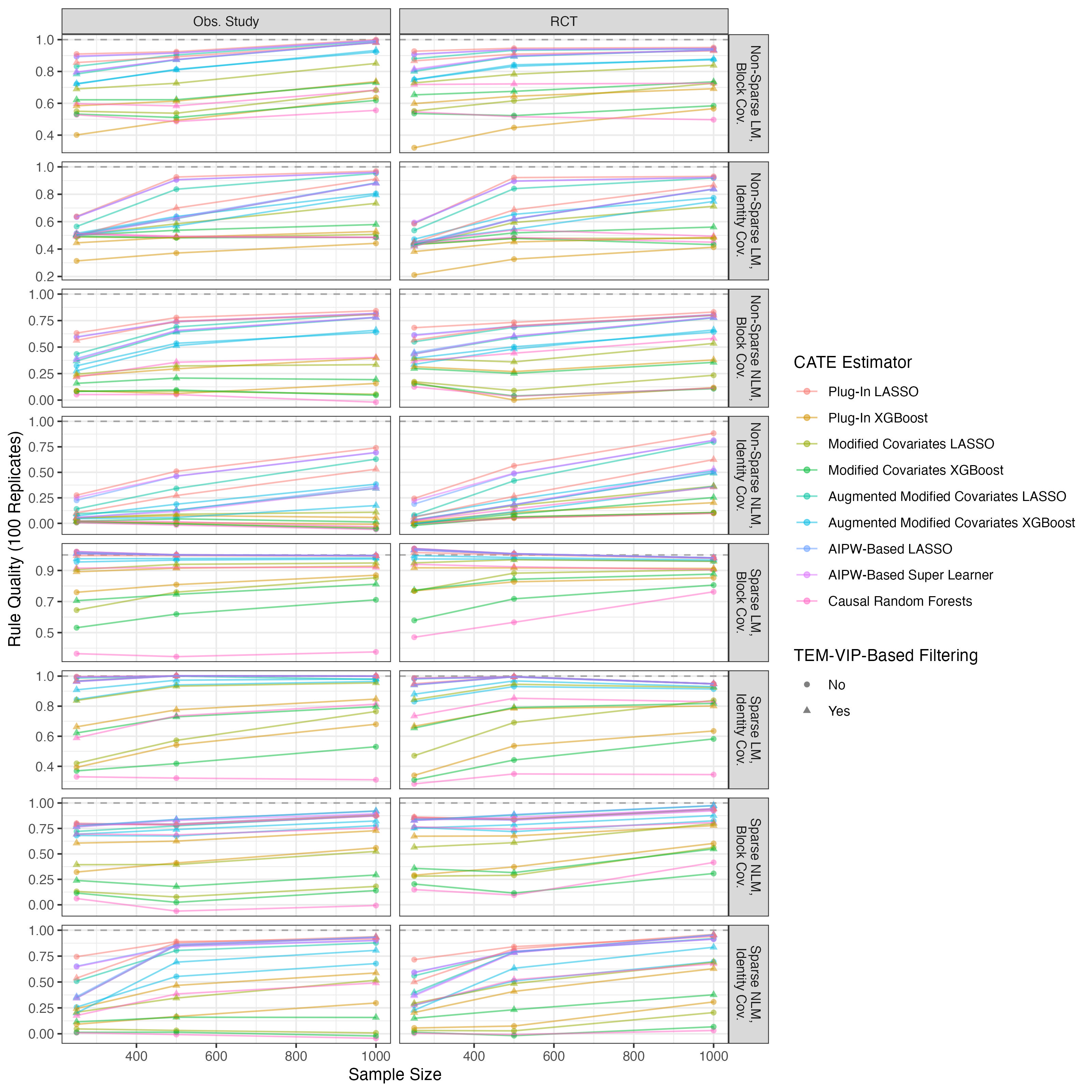}
  \caption{\label{fig:expected-outcome} {\em Rule Quality}: \reviewer{The
      relative rule quality computed over the 100 test set replicates vs. sample
      size, for each combination of DGP. Relative rule quality is defined as the
      mean ITR outcome divided by the optimal ITR defined in
      Proposition~\ref{prop:optimal-itr-mean}, which is approximated using a
      Monte Carlo procedure for each DGP. The dotted line corresponds to the
      idealized relative rule quality.}}
\end{figure}

\begin{figure}
  \centering
  \includegraphics[width=\textwidth]{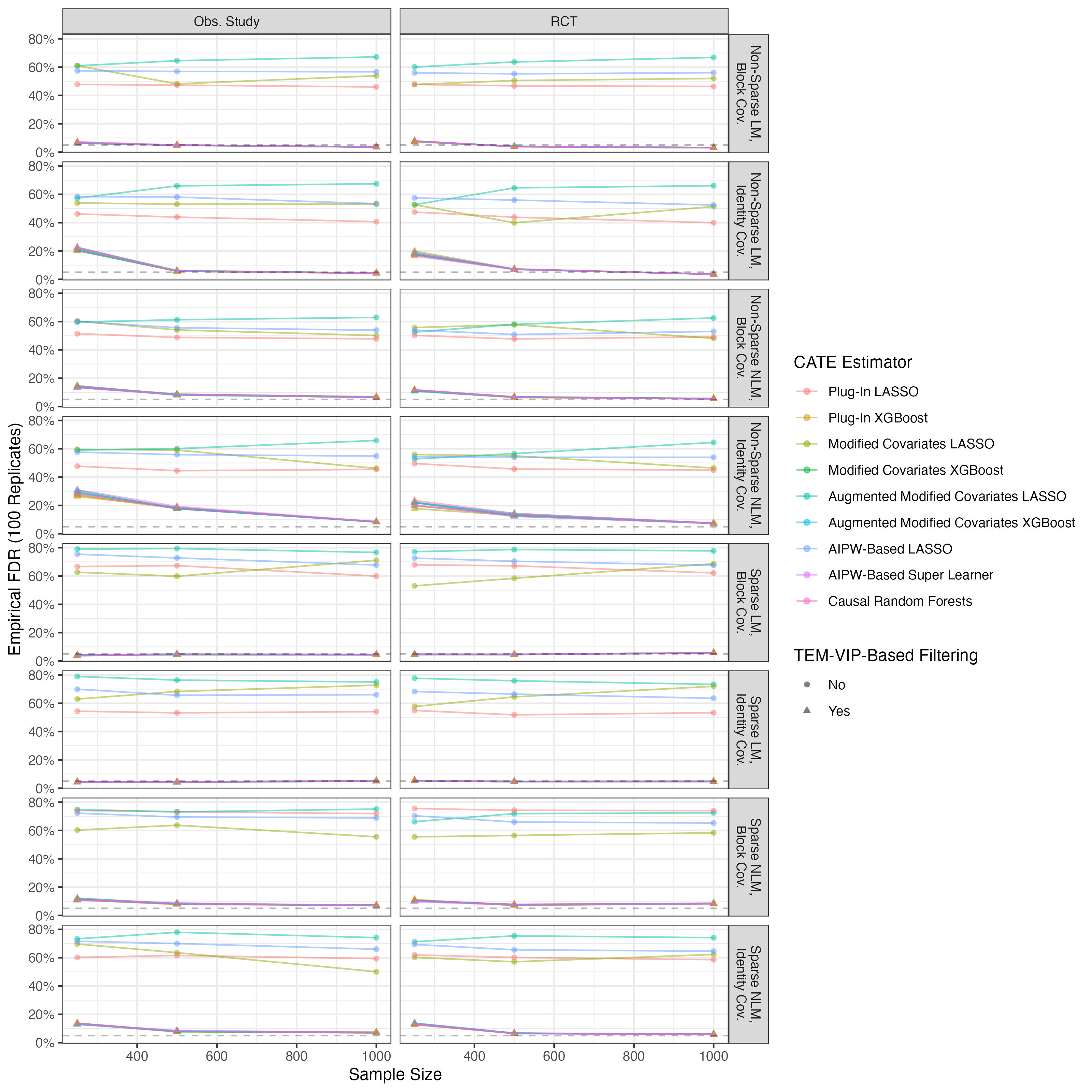}
  \caption{\label{fig:fdr} {\em Accurate Interpretability, FDR:} Empirical FDR
    computed over the 100 learning set replicates vs. sample size for each
    combination of DGP. The dotted line corresponds to desired nominal Type I
    error rate of $5\%$.}
\end{figure}

\begin{figure}
  \centering
  \includegraphics[width=\textwidth]{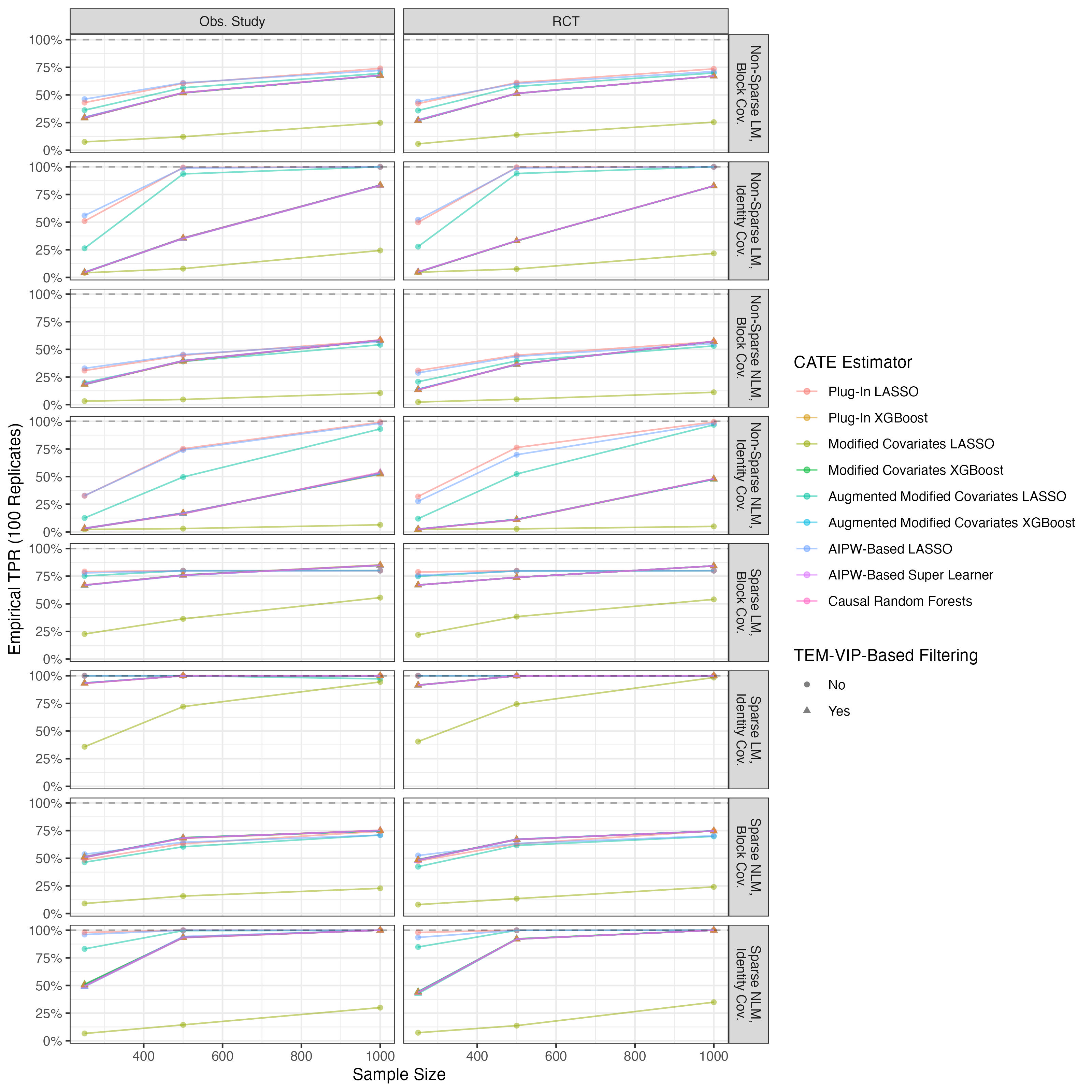}
  \caption{\label{fig:tpr} {\em \reviewer{Accurate} Interpretability, TPR:}
    Empirical TPR computed over the 100 learning set replicates vs. sample size
    for each combination of DGP. The dotted line corresponds to desired TPR of
    $100\%$.}
\end{figure}

\begin{figure}
  \centering
  \includegraphics[width=\textwidth]{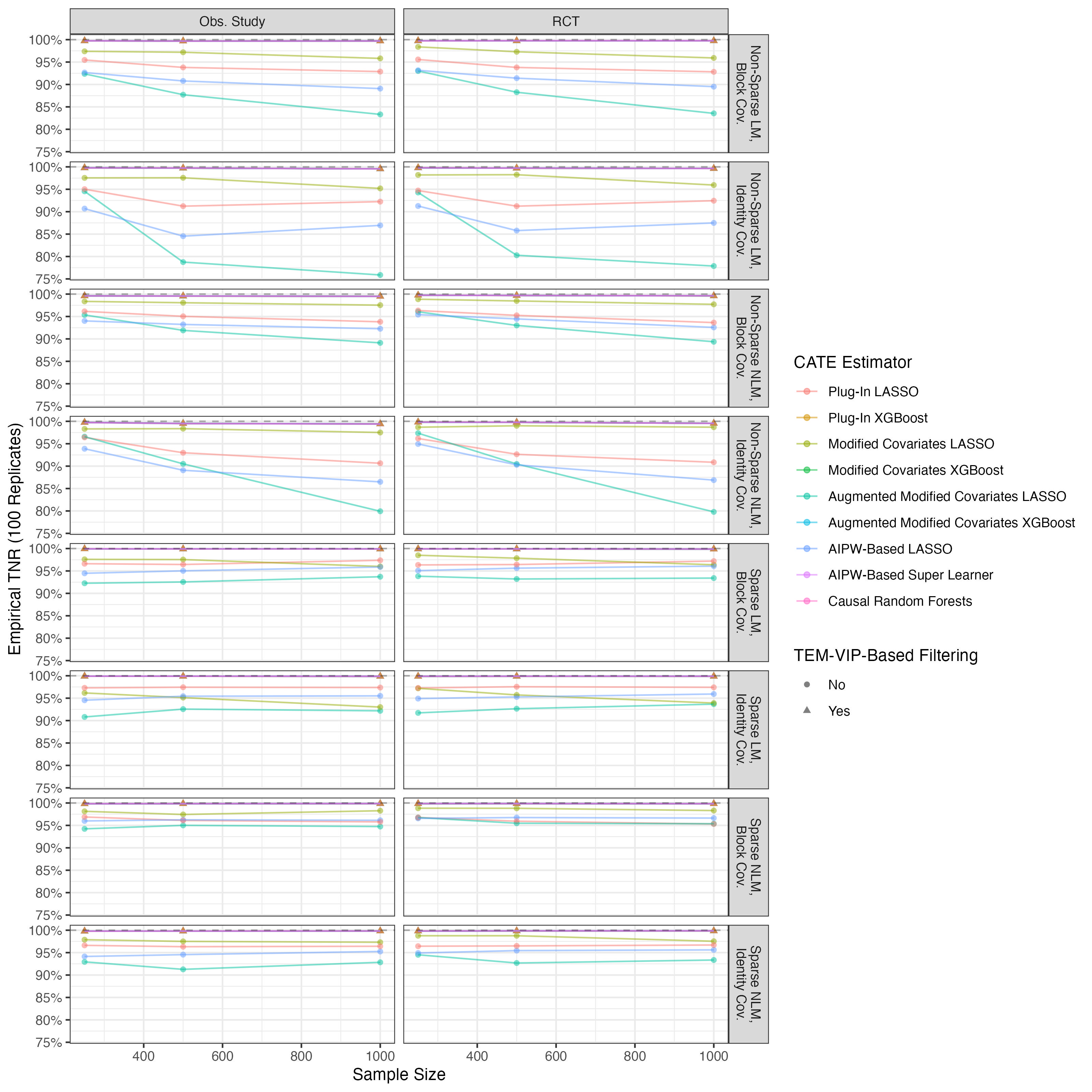}
  \caption{\label{fig:tnr} {\em \reviewer{Accurate} Interpretability, TNR:}
    Empirical TNR computed over the 100 learning set replicates vs. sample size
    for each combination of DGP. The dotted line corresponds to desired TNR of
    $100\%$.}
\end{figure}

\begin{figure}
  \centering
  \includegraphics[width=\textwidth]{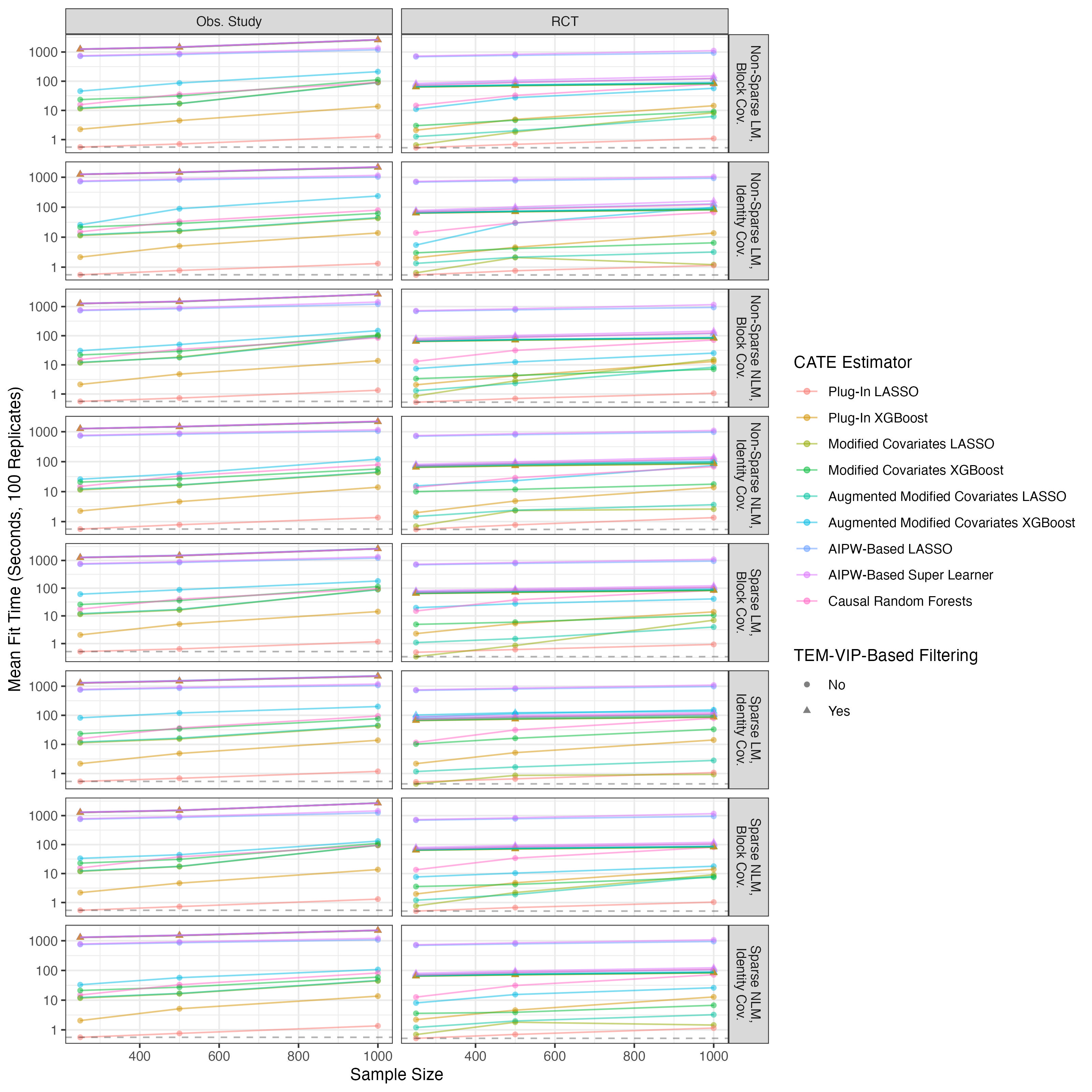}
  \caption{\label{fig:mean-fit-time} {\em Computational Efficiency:} Fit time
    computed over the 100 learning set replicates vs. sample size for each
    combination of DGP. The dotted line corresponds to the minimum mean fit time
    observed in each DGP.}
\end{figure}

\begin{table}[ht]
  \resizebox{\textwidth}{!}{%
\centering
\begin{tabular}{|l||l|rrr|rrr|}
  \hline
  & & \multicolumn{3}{c|}{Unfiltered} & \multicolumn{3}{c|}{TEM-VIP-Based Filtering} \\
CATE Estimator & Metric & n = 250 & n = 500 & n = 1000 & n = 250  & n = 500  & n = 1000  \\
  \hline
  \hline
Plug-In LASSO & Rule quality &  0.99 &  1.00 &  0.95 &  0.95 &  1.00 &  0.95 \\
   & Empirical FDR (\%) & 54.88 & 51.82 & 53.33 &  5.41 &  4.60 &  4.84 \\
   & Empirical TPR (\%) & 100.00 & 100.00 & 100.00 &  91.60 &  99.90 & 100.00 \\
   & Empirical TNR (\%) & 97.28 & 97.56 & 97.42 & 99.88 & 99.89 & 99.89 \\
   & Mean fit time (sec.) &    0.51 &    0.66 &    1.08 &   67.43 &   75.54 &   87.46 \\
  Plug-In XGBoost & Rule quality &  0.34 &  0.54 &  0.63 &  0.67 &  0.79 &  0.80 \\
\hline
\end{tabular}
}
\caption{Simulation results: RCT with Sparse Linear Outcome Model and Identity
  Covariance Matrix}
\end{table}

\begin{table}[ht]
  \resizebox{\textwidth}{!}{%
\centering
\begin{tabular}{|l||l|rrr|rrr|}
  \hline
  & & \multicolumn{3}{c|}{Unfiltered} & \multicolumn{3}{c|}{TEM-VIP-Based Filtering} \\
CATE Estimator & Metric & n = 250 & n = 500 & n = 1000 & n = 250  & n = 500  & n = 1000  \\
  \hline
  \hline
Plug-In LASSO & Rule quality &  1.00 &  1.00 &  1.00 &  0.97 &  1.00 &  1.00 \\
   & Empirical FDR (\%) & 54.39 & 53.24 & 54.06 &  4.58 &  4.57 &  5.25 \\
   & Empirical TPR (\%) & 100.00 & 100.00 & 100.00 &  93.30 & 100.00 & 100.00 \\
   & Empirical TNR (\%) & 97.30 & 97.44 & 97.38 & 99.90 & 99.89 & 99.88 \\
   & Mean fit time (sec.) &    0.54 &    0.69 &    1.19 & 1285.98 & 1498.31 & 2196.98 \\
  Plug-In XGBoost & Rule quality &  0.39 &  0.54 &  0.68 &  0.66 &  0.78 &  0.85 \\
\hline
\end{tabular}}
\caption{Simulation results: Observational Study with Sparse Linear Outcome
  Model and Identity Covariance Matrix}
\end{table}

\begin{table}[ht]
  \resizebox{\textwidth}{!}{%
\centering
\begin{tabular}{|l||l|rrr|rrr|}
  \hline
  & & \multicolumn{3}{c|}{Unfiltered} & \multicolumn{3}{c|}{TEM-VIP-Based Filtering} \\
CATE Estimator & Metric & n = 250 & n = 500 & n = 1000 & n = 250  & n = 500  & n = 1000  \\
  \hline
Plug-In LASSO & Rule quality &  1.04 &  1.01 &  0.98 &  1.02 &  1.00 &  0.98 \\
   & Empirical FDR (\%) & 67.91 & 67.18 & 62.19 &  4.60 &  4.59 &  5.73 \\
   & Empirical TPR (\%) &  78.80 &  80.00 &  80.00 &  67.00 &  73.90 &  84.30 \\
   & Empirical TNR (\%) & 96.34 & 96.43 & 97.20 & 99.92 & 99.91 & 99.88 \\
   & Mean fit time (sec.) &    0.49 &    0.61 &    0.93 &   65.01 &   72.45 &   84.75 \\
  Plug-In XGBoost & Rule quality &  0.77 &  0.83 &  0.85 &  0.92 &  0.92 &  0.91 \\
\hline
\end{tabular}}
\caption{Simulation results: RCT with Sparse Linear Outcome Model and Block
  Covariance Matrix}
\end{table}

\begin{table}[ht]
  \resizebox{\textwidth}{!}{%
\centering
\begin{tabular}{|l||l|rrr|rrr|}
  \hline
  & & \multicolumn{3}{c|}{Unfiltered} & \multicolumn{3}{c|}{TEM-VIP-Based Filtering} \\
CATE Estimator & Metric & n = 250 & n = 500 & n = 1000 & n = 250  & n = 500  & n = 1000  \\
  \hline
Plug-In LASSO & Rule quality &  1.02 &  1.00 &  1.00 &  1.00 &  1.00 &  0.99 \\
   & Empirical FDR (\%) & 66.68 & 67.32 & 59.96 &  4.09 &  4.59 &  4.42 \\
   & Empirical TPR (\%) &  79.20 &  80.00 &  80.00 &  67.00 &  75.90 &  84.90 \\
   & Empirical TNR (\%) & 96.62 & 96.44 & 97.38 & 99.93 & 99.92 & 99.91 \\
   & Mean fit time (sec.) &    0.52 &    0.64 &    1.17 & 1268.36 & 1480.99 & 2614.30 \\
  Plug-In XGBoost & Rule quality &  0.76 &  0.81 &  0.87 &  0.89 &  0.92 &  0.93 \\
\hline
\end{tabular}}
\caption{Simulation results: Observational Study with Sparse Linear Outcome
  Model and Block Covariance Matrix}
\end{table}

\begin{table}[ht]
  \resizebox{\textwidth}{!}{%
\centering
\begin{tabular}{|l||l|rrr|rrr|}
  \hline
  & & \multicolumn{3}{c|}{Unfiltered} & \multicolumn{3}{c|}{TEM-VIP-Based Filtering} \\
CATE Estimator & Metric & n = 250 & n = 500 & n = 1000 & n = 250  & n = 500  & n = 1000  \\
  \hline
Plug-In LASSO & Rule quality &  0.72 &  0.84 &  0.94 &  0.50 &  0.82 &  0.96 \\
   & Empirical FDR (\%) & 61.95 & 60.18 & 58.60 & 12.97 &  6.63 &  5.82 \\
   & Empirical TPR (\%) &  98.00 & 100.00 & 100.00 &  43.50 &  92.00 & 100.00 \\
   & Empirical TNR (\%) & 96.43 & 96.50 & 96.70 & 99.82 & 99.85 & 99.86 \\
   & Mean fit time (sec.) &    0.52 &    0.70 &    1.15 &   65.26 &   72.65 &   84.65 \\
  Plug-In XGBoost & Rule quality &  0.06 &  0.07 &  0.31 &  0.21 &  0.41 &  0.63 \\
\hline
\end{tabular}}
\caption{Simulation results: RCT with Sparse Non-Linear Outcome Model and
  Identity Covariance Matrix}
\end{table}

\begin{table}[ht]
  \resizebox{\textwidth}{!}{%
\centering
\begin{tabular}{|l||l|rrr|rrr|}
  \hline
  & & \multicolumn{3}{c|}{Unfiltered} & \multicolumn{3}{c|}{TEM-VIP-Based Filtering} \\
CATE Estimator & Metric & n = 250 & n = 500 & n = 1000 & n = 250  & n = 500  & n = 1000  \\
  \hline
Plug-In LASSO & Rule quality &  0.74 &  0.89 &  0.92 &  0.54 &  0.88 &  0.94 \\
   & Empirical FDR (\%) & 60.20 & 61.57 & 59.33 & 13.61 &  8.02 &  6.84 \\
   & Empirical TPR (\%) &  98.00 & 100.00 & 100.00 &  50.40 &  93.20 & 100.00 \\
   & Empirical TNR (\%) & 96.62 & 96.32 & 96.42 & 99.79 & 99.81 & 99.84 \\
   & Mean fit time (sec.) &    0.56 &    0.76 &    1.36 & 1288.48 & 1498.63 & 2214.52 \\
  Plug-In XGBoost & Rule quality &  0.09 &  0.17 &  0.30 &  0.24 &  0.47 &  0.59 \\
\hline
\end{tabular}}
\caption{Simulation results: Observational Study with Sparse Non-Linear Outcome
  Model and Identity Covariance Matrix}
\end{table}

\begin{table}[ht]
  \resizebox{\textwidth}{!}{%
\centering
\begin{tabular}{|l||l|rrr|rrr|}
  \hline
  & & \multicolumn{3}{c|}{Unfiltered} & \multicolumn{3}{c|}{TEM-VIP-Based Filtering} \\
CATE Estimator & Metric & n = 250 & n = 500 & n = 1000 & n = 250  & n = 500  & n = 1000  \\
  \hline
Plug-In LASSO & Rule quality &  0.86 &  0.84 &  0.93 &  0.86 &  0.84 &  0.94 \\
   & Empirical FDR (\%) & 75.47 & 74.26 & 74.03 & 10.65 &  7.68 &  8.24 \\
   & Empirical TPR (\%) &  47.50 &  63.00 &  74.80 &  48.80 &  67.20 &  74.70 \\
   & Empirical TNR (\%) & 96.77 & 95.93 & 95.32 & 99.86 & 99.86 & 99.84 \\
   & Mean fit time (sec.) &    0.51 &    0.67 &    1.03 &   64.74 &   72.23 &   85.61 \\
  Plug-In XGBoost & Rule quality &  0.29 &  0.37 &  0.60 &  0.67 &  0.68 &  0.78 \\
\hline
\end{tabular}}
\caption{Simulation results: RCT with Sparse Non-Linear Outcome Model and Block
  Covariance Matrix}
\end{table}

\begin{table}[ht]
  \resizebox{\textwidth}{!}{%
\centering
\begin{tabular}{|l||l|rrr|rrr|}
  \hline
  & & \multicolumn{3}{c|}{Unfiltered} & \multicolumn{3}{c|}{TEM-VIP-Based Filtering} \\
CATE Estimator & Metric & n = 250 & n = 500 & n = 1000 & n = 250  & n = 500  & n = 1000  \\
  \hline
Plug-In LASSO & Rule quality &  0.80 &  0.78 &  0.87 &  0.79 &  0.79 &  0.88 \\
   & Empirical FDR (\%) & 74.10 & 73.25 & 71.86 & 11.16 &  7.73 &  7.34 \\
   & Empirical TPR (\%) &  48.60 &  63.20 &  74.40 &  50.70 &  68.60 &  75.10 \\
   & Empirical TNR (\%) & 96.88 & 96.11 & 95.79 & 99.83 & 99.86 & 99.86 \\
   & Mean fit time (sec.) &    0.54 &    0.72 &    1.31 & 1285.60 & 1501.50 & 2708.36 \\
  Plug-In XGBoost & Rule quality &  0.32 &  0.41 &  0.56 &  0.61 &  0.63 &  0.73 \\
\hline
\end{tabular}}
\caption{Simulation results: Observational Study with Sparse Non-Linear Outcome
  Model and Block Covariance Matrix}
\end{table}

\begin{table}[ht]
  \resizebox{\textwidth}{!}{%
\centering
\begin{tabular}{|l||l|rrr|rrr|}
  \hline
  & & \multicolumn{3}{c|}{Unfiltered} & \multicolumn{3}{c|}{TEM-VIP-Based Filtering} \\
CATE Estimator & Metric & n = 250 & n = 500 & n = 1000 & n = 250  & n = 500  & n = 1000  \\
  \hline
Plug-In LASSO & Rule quality &  0.58 &  0.92 &  0.93 &  0.45 &  0.69 &  0.87 \\
   & Empirical FDR (\%) & 47.50 & 43.85 & 39.97 & 18.48 &  7.14 &  3.58 \\
   & Empirical TPR (\%) &  49.82 &  99.68 & 100.00 &   4.68 &  33.04 &  82.70 \\
   & Empirical TNR (\%) & 94.72 & 91.23 & 92.45 & 99.80 & 99.72 & 99.65 \\
   & Mean fit time (sec.) &    0.55 &    0.76 &    1.15 &   64.45 &   71.90 &   83.55 \\
  Plug-In XGBoost & Rule quality &  0.21 &  0.33 &  0.41 &  0.38 &  0.45 &  0.49 \\
\hline
\end{tabular}}
\caption{Simulation results: RCT with Non-Sparse Linear Outcome Model and
  Identity Covariance Matrix}
\end{table}

\begin{table}[ht]
  \resizebox{\textwidth}{!}{%
\centering
\begin{tabular}{|l||l|rrr|rrr|}
  \hline
  & & \multicolumn{3}{c|}{Unfiltered} & \multicolumn{3}{c|}{TEM-VIP-Based Filtering} \\
CATE Estimator & Metric & n = 250 & n = 500 & n = 1000 & n = 250  & n = 500  & n = 1000  \\
  \hline
Plug-In LASSO & Rule quality &  0.64 &  0.93 &  0.97 &  0.49 &  0.70 &  0.91 \\
   & Empirical FDR (\%) & 46.18 & 43.92 & 40.70 & 20.28 &  5.84 &  4.31 \\
   & Empirical TPR (\%) &  50.94 &  99.42 & 100.00 &   4.26 &  35.76 &  83.50 \\
   & Empirical TNR (\%) & 95.02 & 91.23 & 92.24 & 99.81 & 99.73 & 99.57 \\
   & Mean fit time (sec.) &    0.56 &    0.78 &    1.32 & 1244.32 & 1450.47 & 2129.83 \\
  Plug-In XGBoost & Rule quality &  0.31 &  0.37 &  0.44 &  0.45 &  0.49 &  0.53 \\
\hline
\end{tabular}}
\caption{Simulation results: Observational Study with Non-Sparse Linear Outcome
  Model and Identity Covariance Matrix}
\end{table}

\begin{table}[ht]
  \resizebox{\textwidth}{!}{%
\centering
\begin{tabular}{|l||l|rrr|rrr|}
  \hline
  & & \multicolumn{3}{c|}{Unfiltered} & \multicolumn{3}{c|}{TEM-VIP-Based Filtering} \\
CATE Estimator & Metric & n = 250 & n = 500 & n = 1000 & n = 250  & n = 500  & n = 1000  \\
  \hline
Plug-In LASSO & Rule quality &  0.93 &  0.95 &  0.95 &  0.87 &  0.91 &  0.93 \\
   & Empirical FDR (\%) & 47.78 & 46.78 & 46.36 &  7.62 &  3.96 &  3.12 \\
   & Empirical TPR (\%) &  42.16 &  61.14 &  73.58 &  26.80 &  51.36 &  67.10 \\
   & Empirical TNR (\%) & 95.57 & 93.80 & 92.82 & 99.75 & 99.75 & 99.75 \\
   & Mean fit time (sec.) &    0.53 &    0.69 &    1.09 &   64.13 &   71.02 &   84.58 \\
  Plug-In XGBoost & Rule quality &  0.32 &  0.45 &  0.57 &  0.60 &  0.64 &  0.69 \\
\hline
\end{tabular}}
\caption{Simulation results: RCT with Non-Sparse Linear Outcome Model and Block
  Covariance Matrix}
\end{table}

\begin{table}[ht]
  \resizebox{\textwidth}{!}{%
\centering
\begin{tabular}{|l||l|rrr|rrr|}
  \hline
  & & \multicolumn{3}{c|}{Unfiltered} & \multicolumn{3}{c|}{TEM-VIP-Based Filtering} \\
CATE Estimator & Metric & n = 250 & n = 500 & n = 1000 & n = 250  & n = 500  & n = 1000  \\
  \hline
Plug-In LASSO & Rule quality &  0.91 &  0.92 &  1.00 &  0.86 &  0.89 &  0.98 \\
   & Empirical FDR (\%) & 47.76 & 47.30 & 45.99 &  6.55 &  4.92 &  3.61 \\
   & Empirical TPR (\%) &  43.06 &  60.40 &  73.98 &  29.50 &  51.86 &  67.64 \\
   & Empirical TNR (\%) & 95.46 & 93.80 & 92.87 & 99.75 & 99.68 & 99.71 \\
   & Mean fit time (sec.) &    0.55 &    0.71 &    1.30 & 1245.84 & 1454.06 & 2593.35 \\
  Plug-In XGBoost & Rule quality &  0.40 &  0.49 &  0.64 &  0.58 &  0.61 &  0.74 \\
\hline
\end{tabular}}
\caption{Simulation results: Observational Study with Non-Sparse Linear Outcome
  Model and Block Covariance Matrix}
\end{table}

\begin{table}[ht]
  \resizebox{\textwidth}{!}{%
\centering
\begin{tabular}{|l||l|rrr|rrr|}
  \hline
  & & \multicolumn{3}{c|}{Unfiltered} & \multicolumn{3}{c|}{TEM-VIP-Based Filtering} \\
CATE Estimator & Metric & n = 250 & n = 500 & n = 1000 & n = 250  & n = 500  & n = 1000  \\
  \hline
Plug-In LASSO & Rule quality &  0.24 &  0.56 &  0.88 &  0.06 &  0.27 &  0.62 \\
   & Empirical FDR (\%) & 49.56 & 45.71 & 44.93 & 23.52 & 14.31 &  7.42 \\
   & Empirical TPR (\%) &  32.00 &  76.30 &  99.42 &   2.48 &  11.18 &  47.74 \\
   & Empirical TNR (\%) & 96.16 & 92.64 & 90.86 & 99.79 & 99.76 & 99.56 \\
   & Mean fit time (sec.) &    0.55 &    0.78 &    1.36 &   65.97 &   74.40 &   86.94 \\
  Plug-In XGBoost & Rule quality & -0.01 &  0.05 &  0.10 &  0.02 &  0.11 &  0.20 \\
\hline
\end{tabular}}
\caption{Simulation results: RCT with Non-Sparse Non-Linear Outcome Model and
  Identity Covariance Matrix}
\end{table}

\begin{table}[ht]
  \resizebox{\textwidth}{!}{%
\centering
\begin{tabular}{|l||l|rrr|rrr|}
  \hline
  & & \multicolumn{3}{c|}{Unfiltered} & \multicolumn{3}{c|}{TEM-VIP-Based Filtering} \\
CATE Estimator & Metric & n = 250 & n = 500 & n = 1000 & n = 250  & n = 500  & n = 1000  \\
  \hline
Plug-In LASSO & Rule quality &  0.28 &  0.51 &  0.74 &  0.10 &  0.27 &  0.53 \\
   & Empirical FDR (\%) & 47.71 & 44.51 & 45.48 & 27.09 & 17.74 &  8.42 \\
   & Empirical TPR (\%) &  32.66 &  75.22 &  99.26 &   3.12 &  17.14 &  52.74 \\
   & Empirical TNR (\%) & 96.41 & 92.99 & 90.64 & 99.77 & 99.52 & 99.42 \\
   & Mean fit time (sec.) &    0.56 &    0.79 &    1.37 & 1264.10 & 1464.83 & 2152.27 \\
  Plug-In XGBoost & Rule quality &  0.02 &  0.01 & -0.01 &  0.05 &  0.08 &  0.06 \\
\hline
\end{tabular}}
\caption{Simulation results: Observational Study with Non-Sparse Non-Linear
  Outcome Model and Identity Covariance Matrix}
\end{table}

\begin{table}[ht]
  \resizebox{\textwidth}{!}{%
\centering
\begin{tabular}{|l||l|rrr|rrr|}
  \hline
  & & \multicolumn{3}{c|}{Unfiltered} & \multicolumn{3}{c|}{TEM-VIP-Based Filtering} \\
CATE Estimator & Metric & n = 250 & n = 500 & n = 1000 & n = 250  & n = 500  & n = 1000  \\
  \hline
Plug-In LASSO & Rule quality &  0.68 &  0.73 &  0.83 &  0.57 &  0.70 &  0.80 \\
   & Empirical FDR (\%) & 50.24 & 47.74 & 49.22 & 11.53 &  6.96 &  5.59 \\
   & Empirical TPR (\%) &  30.84 &  44.68 &  57.02 &  13.72 &  36.56 &  57.18 \\
   & Empirical TNR (\%) & 96.33 & 95.26 & 93.64 & 99.76 & 99.67 & 99.61 \\
   & Mean fit time (sec.) &    0.52 &    0.70 &    1.06 &   64.46 &   72.09 &   84.93 \\
  Plug-In XGBoost & Rule quality &  0.16 &  0.00 &  0.12 &  0.32 &  0.27 &  0.38 \\
\hline
\end{tabular}}
\caption{Simulation results: RCT with Non-Sparse Non-Linear Outcome Model and
  Block Covariance Matrix}
\end{table}

\begin{table}[ht]
  \resizebox{\textwidth}{!}{%
\centering
\begin{tabular}{|l||l|rrr|rrr|}
  \hline
  & & \multicolumn{3}{c|}{Unfiltered} & \multicolumn{3}{c|}{TEM-VIP-Based Filtering} \\
CATE Estimator & Metric & n = 250 & n = 500 & n = 1000 & n = 250  & n = 500  & n = 1000  \\
  \hline
Plug-In LASSO & Rule quality &  0.63 &  0.78 &  0.84 &  0.56 &  0.74 &  0.81 \\
   & Empirical FDR (\%) & 51.46 & 48.81 & 47.81 & 13.67 &  8.29 &  7.00 \\
   & Empirical TPR (\%) &  30.84 &  44.66 &  58.12 &  18.72 &  39.18 &  58.44 \\
   & Empirical TNR (\%) & 96.15 & 95.04 & 93.82 & 99.61 & 99.58 & 99.48 \\
   & Mean fit time (sec.) &    0.56 &    0.73 &    1.34 & 1266.36 & 1470.17 & 2634.34 \\
  Plug-In XGBoost & Rule quality &  0.09 &  0.06 &  0.16 &  0.23 &  0.29 &  0.40 \\
\hline
\end{tabular}}
\caption{Simulation results: Observational Study with Non-Sparse Non-Linear
  Outcome Model and Block Covariance Matrix}
\end{table}

\end{document}